\documentclass[letterpaper, 10 pt, conference]{ieeeconf}
\IEEEoverridecommandlockouts
\usepackage{amsmath,amssymb,amsfonts}
\usepackage[ruled,linesnumbered]{algorithm2e}
\usepackage{array}
\usepackage{cite}
\usepackage{color}
\usepackage{cuted}
\usepackage{graphicx}
\usepackage{hyperref}
\usepackage{mathrsfs}
\usepackage{multirow}
\usepackage{scalerel}
\usepackage[caption=false]{subfig}
\usepackage{verbatim}

\newcommand\dif{\mathrm{d}}

\newtheorem{theorem}{\textbf{Theorem}}

\title{\LARGE \bf Optimal Investment under Mutual Decision Influence among Agents}

\author{Huisheng Wang and H. Vicky Zhao
\thanks{Huisheng Wang and H. Vicky Zhao are with the Department of Automation, Tsinghua University, Beijing, 100084, China. Emails: {\tt\small \url{whs22@mails.tsinghua.edu.cn}, \url{vzhao@tsinghua.edu.cn}}}
}

\begin{document}

\maketitle
\thispagestyle{empty}
\pagestyle{empty}

\begin{abstract}
In financial markets, agents often mutually influence each other's investment strategies and adjust their strategies to align with others. However, there is limited quantitative study of agents' investment strategies in such scenarios. In this work, we formulate the optimal investment differential game problem to study the mutual influence among agents. We derive the analytical solutions for agents' optimal strategies and propose a fast algorithm to find approximate solutions with low computational complexity. We theoretically analyze the impact of mutual influence on agents' optimal strategies and terminal wealth. When the mutual influence is strong and approaches infinity, we show that agents' optimal strategies converge to the asymptotic strategy. Furthermore, in general cases, we prove that agents' optimal strategies are linear combinations of the asymptotic strategy and their rational strategies without others' influence. We validate the performance of the fast algorithm and verify the correctness of our analysis using numerical experiments. This work is crucial to comprehend mutual influence among agents and design effective mechanisms to guide their strategies in financial markets.
\end{abstract}


\section{Introduction}
\label{sec:intro}
Agents, including investors and companies, often design their portfolios based on risk preferences to achieve higher returns and lower volatility and maximize their utility \cite{merton1969lifetime, rogers2013optimal, leung2021optimal, moehle2021certainty}. With the rise of online social networks, investment forum platforms have become useful tools, enabling agents to follow each other and share their investment strategies \cite{hao2019online}. Therefore, agents are often \textit{mutually} influenced by each other when making investment strategies, 
and adjust their strategies to align with others \cite{delellis2017evolution}. While moderate alignment with others' strategies can be beneficial and provide valuable insights, excessive mutual influence often leads to agents' suboptimal investment strategies \cite{ahmad2022does}. Therefore, it is crucial to analyze the impact of mutual influence on agents' investment strategies in social networks and design mechanisms to mitigate the adverse effects of excessive mutual influence.

Many studies have examined others' influence on agents' investment strategy-making processes. 
The works in \cite{chang2000examination, yu2019detection} qualitatively confirmed the significant impact of others' influence on agents' investment strategies in equity and future markets using empirical analysis. Our prior work in \cite{wang2024herd} formulated an optimal investment problem involving one leading expert and one normal agent, where the leading expert unilaterally influences the normal agent but not vice versa. 
We derived the normal agent's optimal investment strategies using optimal control theory, and quantitatively analyzed the influence of the leading expert on the normal agent's optimal investment strategies in \cite{wang2024herd}. However, the above studies do not consider the scenario where multiple agents \textit{mutually} influence each other's investment strategies.

When agents mutually influence each other, it forms a \textit{multi-agent differential game problem}. In this problem, we assume that agents make rational strategies, and they are aware that others also make rational strategies. Furthermore, they are aware that others are aware of their rationality when making strategies, and so on, creating a recursive awareness of rationality \cite{friedman2013differential}. 
An approach to solve the multi-agent differential game problem is the distributed method, which determines the optimal strategy for each agent and then jointly derives their optimal strategies \cite{suzuki2014game}. However, this approach has high computational complexity, especially when the number of agents is large. Alternatively, the multi-agent differential game problem can be tackled using reinforcement learning \cite{zhang2021multi, zhang2024optimizing}. 
Nevertheless, this approach lacks analytical solutions for agents' optimal strategies, which poses challenges for theoretical analysis of the impact of mutual influence on agents' optimal strategies.

In this work, we study the optimal investment under mutual strategy influence among agents. Our contributions can be summarized as follows:
\begin{itemize}
    \item We formulate an optimal investment differential game problem to model the mutual influence among agents.
    \item We derive the analytical solutions for agents' optimal strategies and propose a fast algorithm to find approximated solutions with low computational complexity.
    \item We theoretically analyze the impact of mutual influence on agents' optimal strategies and terminal wealth. 
    \item We conduct numerical experiments to validate the accuracy and efficiency of our proposed fast algorithm and verify the correctness of our analysis.
\end{itemize}

The rest of this paper is organized as follows. In Section \ref{sec:formulation}, we formulate the optimal investment differential game problem. In Section \ref{sec:solution}, we derive the analytical solutions of agents' optimal strategies and propose the fast algorithm. In Section \ref{sec:analysis}, we derive the asymptotic strategy and theoretically analyze the impact of mutual influence on agents' optimal strategies and terminal wealth. The experimental results are in Section \ref{sec:experiment}. Section \ref{sec:conclusion} is the conclusion.

\section{Problem Formulation}
\label{sec:formulation}
We consider a financial market with $m$ risky assets and $n$ agents. We denote the set of risky assets as $\mathcal{M}:=\{1,2,\ldots,m\}$, the set of agents as $\mathcal{N}:=\{1,2,\ldots,n\}$, and the investment period as $\mathcal{T}:=[0, T]$. To model the randomness of the risky asset prices, let $\{\boldsymbol{B}(t)\}_{t\in\mathcal{T}}$ be an $m$-dimensional standard Brownian motion, and we assume that the $m$ components of $\{\boldsymbol{B}(t)\}_{t\in\mathcal{T}}$ are independent of each other \cite{merton1969lifetime}.
Following our prior work in \cite{wang2024herd}, we define the wealth invested in the risky assets by the $j$-th agent at time $t$ as his/her \textit{strategy}, denoted by $\boldsymbol{P}_j(t)\in\mathbf{R}^m$. Given the interest rate $r$, the excess return rates $\boldsymbol{\nu}\in\mathbf{R}^m$, and the volatility $\boldsymbol{\sigma}\in\mathbf{R}^{m\times m}$ of the risky assets, the $j$-th agent's wealth process $\{X_j(t)\}_{t\in\mathcal{T}}$ satisfies
\begin{equation}
    \dif X_j(t)=rX_j(t)\dif t+\boldsymbol{P}_j^\top(t)[\boldsymbol{\nu}\dif t+\boldsymbol{\sigma}\dif\boldsymbol{B}(t)],
    \label{eq:wealth}
\end{equation}
where $X_j(0)=x_j$ is his/her initial wealth. To simplify the notation, we denote $\boldsymbol{\varSigma}:=\boldsymbol{\sigma}\boldsymbol{\sigma}^\top$ as the covariance matrix of the risky assets' prices and $\kappa:=\boldsymbol{\nu}^\top\boldsymbol{\varSigma}^{-1}\boldsymbol{\nu}$. 


Considering the mutual influence among agents, the $j$-th agent jointly maximizes his/her expected utility of the terminal wealth $\mathbb{E}\phi_j[X_j(T)]$ and minimizes the distance between his/her own and others' strategies. The $j$-th agent's utility of the terminal wealth is defined as
\begin{equation}
    \phi_j[X_j(T)]:=-\alpha_j^{-1}\exp[-\alpha_jX_j(T)],
\end{equation}
where $\alpha_j$ is his/her risk aversion coefficient \cite{pratt1978risk}. 

To model the mutual influence, we denote the adjacency matrix among the $n$ agents as $\boldsymbol{W}=[w_{ji}]_{j,i\in\mathcal{N}}\in\mathbf{R}^{n\times n}$, whose element $w_{ji}$ quantifies the influence of the $i$-th agent's strategy on the $j$-th agent's strategy. We assume that $\boldsymbol{W}$ is row-stochastic and satisfies $w_{ji}\geqslant0$ and $\sum_{i\in\mathcal{N}}w_{ji}=1$ for all $j,i\in\mathcal{N}$. Given 
$\boldsymbol{W}$, the $j$-th agent's strategy $\boldsymbol{P}_j(t)$ at time $t$ tends to align with his/her average followees' strategy $\boldsymbol{Q}_j(t):=\sum_{i\in\mathcal{N}}w_{ji}\boldsymbol{P}_i(t)$. Following our prior work in \cite{wang2024herd}, we use the \textit{average deviation} to measure the distance between the two strategies, which is
\begin{equation}
    D(\boldsymbol{P}_j,\boldsymbol{Q}_j):=\frac{1}{2}\int_0^T\mathrm{e}^{2r(T-t)}(\boldsymbol{P}_j-\boldsymbol{Q}_j)^\top(\boldsymbol{P}_j-\boldsymbol{Q}_j)(t)\dif t,
\end{equation}
where the exponential decay term $\mathrm{e}^{2r(T-t)}$ implies that deviations occurring later carry less weight \cite{wang2024herd}.

In summary, the proposed optimal investment differential game problem is: for all $j\in\mathcal{N}$,
\begin{equation}\left\{
    \begin{aligned}
        &\sup_{\{\boldsymbol{P}_j(t)\}_{t\in\mathcal{T}}}\mathbb{E}\phi_j[X_j(T)]-\theta_jD(\boldsymbol{P}_j,\boldsymbol{Q}_j),\\
        &\dif X_j(t)=rX_j(t)\dif t+\boldsymbol{P}_j^\top(t)[\boldsymbol{\nu}\dif t+\boldsymbol{\sigma}\dif\boldsymbol{B}(t)],
    \end{aligned}\right.
    \label{eq:game_problem}
\end{equation}
where the influence coefficient $\{\theta_j\}_{j\in\mathcal{N}}$ is to address the tradeoff between the two different objectives, i.e., maximizing 
$\mathbb{E}\phi_j[X_j(T)]$ and minimizing 
$D(\boldsymbol{P}_j,\boldsymbol{Q}_j)$ \cite{wang2024herd}. 
When $\theta_j=0$, the $j$-th agent's optimal strategy $\{\boldsymbol{P}_j^*(t)\}_{t\in\mathcal{T}}$ is independent of others' strategies, which we call the \textit{rational strategy}. From \cite{merton1969lifetime}, the $j$-th agent's rational strategy $\{\bar{\boldsymbol{P}}_j(t)\}_{t\in\mathcal{T}}$ is
\begin{equation}
    \bar{\boldsymbol{P}}_j(t)=\alpha_j^{-1}\mathrm{e}^{r(t-T)}\boldsymbol{\varSigma}^{-1}\boldsymbol{\nu}.
    \label{eq:rational}
\end{equation}

\section{Optimal Investment Strategies}
\label{sec:solution}
In this section, we solve the problem \eqref{eq:game_problem} to obtain agents' optimal investment strategies. We show that agents' optimal strategies are linear combinations of their rational strategies, and propose a fast algorithm to find the approximated optimal strategies with low computational complexity. 

\subsection{The Analytical Solution of Agents' Optimal Strategies}
We first show that the $j$-th agent's optimal strategy $\{\boldsymbol{P}_j^*(t)\}_{t\in\mathcal{T}}$ is a linear combination of his/her rational strategy $\{\bar{\boldsymbol{P}}_j(t)\}_{t\in\mathcal{T}}$ and his/her average followees' optimal strategy $\{\boldsymbol{Q}_j^*(t)\}_{t\in\mathcal{T}}$. In the following, we denote $\boldsymbol{I}_m\in\mathbb{R}^{m\times m}$ and $\boldsymbol{O}_m\in\mathbb{R}^{m\times m}$ as the identity and all-zero matrices, respectively.

\begin{theorem}\label{the:1}
The $j$-th agent's optimal strategy $\{\boldsymbol{P}_j^*(t)\}_{t\in\mathcal{T}}$ in the differential game problem \eqref{eq:game_problem} is
\begin{equation}
    \boldsymbol{P}_j^*(t)=\boldsymbol{Z}_j\bar{\boldsymbol{P}}_j(t)+(\boldsymbol{I}_m-\boldsymbol{Z}_j)\boldsymbol{Q}_j^*(t),
    \label{eq:decomposition1}
\end{equation}
where $\boldsymbol{Q}^*_j(t)$ is the average followees' optimal strategy with
\begin{equation}
    \boldsymbol{Q}^*_j(t):=\sum_{i\in\mathcal{N}}w_{ji}\boldsymbol{P}_i^*(t),
    \label{eq:q}
\end{equation}
$\boldsymbol{Z}_j$ is called the $j$-th agent's investment opinion with
\begin{equation}
    \boldsymbol{Z}_j:=(\boldsymbol{I}_m+\alpha_j^{-1}\eta_j^{-1}\theta_j\boldsymbol{\varSigma}^{-1})^{-1},
    \label{eq:opinion}
\end{equation}
and $\eta_j$ is the integral constant with
\begin{align}
     \eta_j&:=\exp\left[-\alpha_jx_j\mathrm{e}^{rT}-\alpha_j\int_0^T\mathrm{e}^{r(T-t)}\boldsymbol{\nu}^\top\boldsymbol{P}_j^*(t)\dif t\right.\notag\\
     &\left.+\frac{\alpha_j^2}{2}\int_0^T\mathrm{e}^{2r(T-t)}\boldsymbol{P}_j^{*\top}(t)\boldsymbol{\varSigma}\boldsymbol{P}_j^*(t)\dif t\right].
     \label{eq:integral}
\end{align}
\end{theorem}

\begin{proof}
The proof is similar to the proof of Theorem 1 in our prior work in \cite{wang2024herd} and omitted here. 
\end{proof}

Following our prior work in \cite{wang2024herd}, we call the weight matrix $\boldsymbol{Z}_j$ the $j$-th agent's investment opinion, which quantifies the $j$-th agent's adherence to his/her rational strategy in his/her optimal strategy $\{\boldsymbol{P}_j^*(t)\}_{t\in\mathcal{T}}$. 

Next, we further show that the $j$-th agent's optimal strategy $\{\boldsymbol{P}_j^*(t)\}_{t\in\mathcal{T}}$ is a linear combination of all $n$ agents' rational strategies. To simplify the notation, we denote $\boldsymbol{U}:=[\boldsymbol{U}_{ji}]_{j,i\in\mathcal{N}}\in\mathbf{R}^{mn\times mn}$, $\boldsymbol{Z}:=\mathrm{diag}(\boldsymbol{Z}_j)_{j\in\mathcal{N}}\in\mathbf{R}^{mn\times mn}$,
$\boldsymbol{S}:=\boldsymbol{I}_{mn}-\boldsymbol{Z}\in\mathbf{R}^{mn\times mn}$, $\boldsymbol{P}^*(t):=[\boldsymbol{P}_j^*(t)]_{j\in\mathcal{N}}^\top\in\mathbf{R}^{mn}$, and $\bar{\boldsymbol{P}}(t):=[\bar{\boldsymbol{P}}_j(t)]_{j\in\mathcal{N}}^\top\in\mathbf{R}^{mn}$ in the concatenated form.

\begin{theorem}\label{the:2}
The $j$-th agent's optimal strategy $\{\boldsymbol{P}_j^*(t)\}_{t\in\mathcal{T}}$ in the differential game problem \eqref{eq:game_problem} is 
\begin{equation}
    \boldsymbol{P}_j^*(t)=\sum_{i\in\mathcal{N}}\boldsymbol{U}_{ji}\bar{\boldsymbol{P}}_i(t),
    \label{eq:decomposition}
\end{equation}
where $\boldsymbol{U}_{ji},i\in\mathcal{N}$ are elements in the $j$-th row of the matrix
\begin{equation}
    \boldsymbol{U}=[\boldsymbol{I}_{mn}-\boldsymbol{S}(\boldsymbol{W}\otimes\boldsymbol{I}_m)]^{-1}\boldsymbol{Z},
    \label{eq:u}
\end{equation}
and satisfy
\begin{equation}
    \sum_{i\in\mathcal{N}}\boldsymbol{U}_{ji}=\boldsymbol{I}_m.
    \label{eq:sumu}
\end{equation}
Here, $\otimes$ represents the Kronecker product of two matrices.
\end{theorem}

\begin{proof}
From \eqref{eq:decomposition1}, we have
\begin{equation}
    \boldsymbol{P}^*(t)=\boldsymbol{Z}\bar{\boldsymbol{P}}(t)+\boldsymbol{S}(\boldsymbol{W}\otimes\boldsymbol{I}_m)\boldsymbol{P}^*(t).
    \label{eq:ZS_concate}
\end{equation}
Therefore, using matrix operations, we can prove that
\begin{equation}
    \boldsymbol{P}^*(t)=\boldsymbol{U}\bar{\boldsymbol{P}}(t),
    \label{eq:pu}
\end{equation}
where $\boldsymbol{U}$ is in \eqref{eq:u}. When writing \eqref{eq:pu} in the component form, we have \eqref{eq:decomposition}. Denote $\boldsymbol{I}:=\boldsymbol{1}_n\otimes\boldsymbol{I}_m\in\mathbf{R}^{mn\times m}$ and $\boldsymbol{1}_n\in\mathbf{R}^n$ as the all-one vector, and we can prove that
\begin{align}
    &[\boldsymbol{I}_{mn}-\boldsymbol{S}(\boldsymbol{W}\otimes\boldsymbol{I}_m)]\boldsymbol{I}=\boldsymbol{I}-\boldsymbol{S}(\boldsymbol{W}\otimes\boldsymbol{I}_m)(\boldsymbol{1}_n\otimes\boldsymbol{I}_m)\notag\\
    =\ &\boldsymbol{I}-\boldsymbol{S}(\boldsymbol{W}\boldsymbol{1}_n)\otimes(\boldsymbol{I}_m\boldsymbol{I}_m)=\boldsymbol{I}-\boldsymbol{S}(\boldsymbol{1}_n\otimes\boldsymbol{I}_m)\notag\\
    =\ &(\boldsymbol{I}_{mn}-\boldsymbol{S})\boldsymbol{I}=\boldsymbol{Z}\boldsymbol{I},
    \label{eq:ZI}
\end{align}
where we use the property of the row-stochastic adjacency matrix $\boldsymbol{W}\boldsymbol{1}_n=\boldsymbol{1}_n$, and the property of the Kronecker product $(\boldsymbol{A}\otimes\boldsymbol{B})(\boldsymbol{C}\otimes\boldsymbol{D})=(\boldsymbol{A}\boldsymbol{C})\otimes(\boldsymbol{B}\boldsymbol{D})$ when the dimensions of $\boldsymbol{A}$, $\boldsymbol{B}$, $\boldsymbol{C}$, and $\boldsymbol{D}$ are compatible for multiplication. Therefore, from \eqref{eq:u} and \eqref{eq:ZI}, we have
\begin{align}
    \boldsymbol{U}\boldsymbol{I}=[\boldsymbol{I}_{mn}-\boldsymbol{S}(\boldsymbol{W}\otimes\boldsymbol{I}_m)]^{-1}\boldsymbol{Z}\boldsymbol{I}=\boldsymbol{I}.
    \label{eq:ui=i}
\end{align}
When writing \eqref{eq:ui=i} in the component form, we have \eqref{eq:sumu}.
\end{proof}

\subsection{The Proposed Fast Algorithm}
From Theorem \ref{the:1} and Theorem \ref{the:2}, we can obtain the analytical solutions for agents' optimal strategies. However, this method has high computational complexity. Especially, in \eqref{eq:integral}, we must compute the integral $\{\eta_j\}_{j\in\mathcal{N}}$ involving vectors. In Python, we can use \texttt{scipy.integrate.quad} to compute $\{\eta_j\}_{j\in\mathcal{N}}$ with a small error tolerance of $1.49\times10^{-8}$, which is very close to the exact value. Since \texttt{scipy.integrate.quad} uses adaptive quadrature techniques, it evaluates the integrand at numerous sample points and sums the results, which makes computing the integral with such high accuracy computationally expensive. To address this problem, we propose a fast algorithm to approximate $\{\eta_j\}_{j\in\mathcal{N}}$, as shown in Algorithm \ref{alg:1}.

In the fast algorithm, we utilize the right rectangle method in \cite{davis2007methods} and approximate $\eta_j$ with
\begin{equation}
     \hat{\eta}_j=\zeta_j\exp\left[-\alpha_jT\boldsymbol{\nu}^\top\boldsymbol{P}_j^*(T)+\frac{\alpha_j^2T}{2}\boldsymbol{P}_j^{*\top}(T)\boldsymbol{\varSigma}\boldsymbol{P}_j^*(T)\right],
     \label{eq:hatintegral}
\end{equation}
where $\zeta_j:=\exp(-\alpha_jx_j\mathrm{e}^{rT})$. In \eqref{eq:hatintegral}, we only need to perform simple matrix operations, which reduces the computational complexity. 
Compared to \texttt{scipy.integrate.quad}, the fast algorithm has a larger error when computing $\{\eta_j\}_{j\in\mathcal{N}}$, while the overall error when calculating agents' optimal strategies $\{\boldsymbol{P}_j^*(t)\}_{t\in\mathcal{T}}$ remains small, which will be validated in Section \ref{sec:experiment}.

\begin{algorithm}[t]\label{alg:1}
\small
	\caption{The Fast Algorithm.}
	\KwData{Interest rate $r$, excess return rate $\boldsymbol{\nu}$, volatility $\boldsymbol{\sigma}$, investment period $T$, adjacency matrix $\boldsymbol{W}$, risk aversion coefficients $\{\alpha_j\}_{j\in\mathcal{N}}$, influence coefficients $\{\theta_j\}_{j\in\mathcal{N}}$, initial wealth $\{x_j\}_{j\in\mathcal{N}}$, error tolerance $\varepsilon$.}
	\KwResult{Optimal strategies $\{\boldsymbol{P}_j^*(t)\}_{t\in\mathcal{T}},j\in\mathcal{N}$.}  
	\BlankLine
	$k=0,\delta\eta^{(0)}=+\infty,\eta_j^{(0)}=\exp(-\alpha_jx_j\mathrm{e}^{rT}-\frac{1}{2}\kappa T), j\in\mathcal{N}$;

	\While{$\delta\eta^{(k)}\geqslant\varepsilon$}{
  Calculate $\boldsymbol{Z}_j^{(k)},j\in\mathcal{N}$ using \eqref{eq:opinion};

  Calculate $\boldsymbol{U}_{ji}^{(k)},j,i\in\mathcal{N}$ using \eqref{eq:u};

  Calculate $\{\boldsymbol{P}_j^{(k)}(t)\}_{t\in\mathcal{T}},j\in\mathcal{N}$ using \eqref{eq:decomposition};
  
  Calculate $\eta_j^{(k+1)},j\in\mathcal{N}$ using \eqref{eq:hatintegral};
  
  $\delta\eta^{(k+1)}=\max_{j\in\mathcal{N}}|\eta_j^{(k+1)}-\eta_j^{(k)}|$;}
  
  $\boldsymbol{P}_j^*(t)\approx\boldsymbol{P}_j^{(k)}(t),t\in\mathcal{T},j\in\mathcal{N}$.
\end{algorithm}

\section{Influence of Mutual Influence on Agents' Optimal Strategies and Terminal Wealth}
\label{sec:analysis}
In this section, we theoretically analyze the impact of the mutual influence on agents' optimal strategies and their terminal wealth. To facilitate theoretical analysis, we consider the simple scenario where the influence of all $n$ agents on the $j$-th agent is homogeneous, i.e., the adjacency matrix $\boldsymbol{W}:=\frac{1}{n}\boldsymbol{1}_n\boldsymbol{1}_n^\top$. We will study the general cases with arbitrary $\boldsymbol{W}$ using numerical experiments in Section \ref{sec:experiment}. 

\subsection{The Optimal Strategies}
To analyze the impact of mutual influence on agents' optimal strategies, we first derive the asymptotic strategy as the influence coefficients $\{\theta_j\}_{j\in\mathcal{N}}$ approach infinity, and then show that agents' optimal strategies are linear combinations of their rational strategies and the asymptotic strategy.

\subsubsection{The Asymptotic Strategy}
First, we can show that, as the influence coefficients $\{\theta_j\}_{j\in\mathcal{N}}$ approach infinity, the optimal strategies of all $n$ agents converge to the same function, which we call the \textit{asymptotic strategy}. Thus, analyzing the impact of the mutual influence on the $j$-th agent's optimal strategy is equivalent to comparing his/her rational strategy $\{\bar{\boldsymbol{P}}_j(t)\}_{t\in\mathcal{T}}$ and the asymptotic strategy $\{\tilde{\boldsymbol{P}}(t)\}_{t\in\mathcal{T}}$.

\begin{theorem}\label{the:3}
As the influence coefficients $\{\theta_j\}_{j\in\mathcal{N}}$ approach infinity, all $n$ agents' optimal strategies converge to 
\begin{equation}
    \tilde{\boldsymbol{P}}(t)=\tilde{\alpha}^{-1}\mathrm{e}^{r(t-T)}\boldsymbol{\varSigma}^{-1}\boldsymbol{\nu},
    \label{eq:asymptotic}
\end{equation}
where $\tilde{\alpha}$ is the \textit{asymptotic risk aversion coefficient} with
\begin{equation}
    \tilde{\alpha}:=\frac{\sum_{j\in\mathcal{N}}\eta_j\alpha_j}{\sum_{j\in\mathcal{N}}\eta_j}.
    \label{eq:asy_alpha}
\end{equation}
\end{theorem}

\begin{proof}
First, we prove that when $\{\theta_j\}_{j\in\mathcal{N}}$ approaches infinity, all $n$ agents' optimal strategies converge to the same function. We prove it by contradiction. Assume there exist $j',j''\in\mathcal{N}$ and $t\in\mathcal{T}$ such that $\boldsymbol{P}_{j'}^*(t)\ne\boldsymbol{P}_{j''}^*(t)$, and we can find at least one agent, e.g., the $j$-th agent, whose average deviation $D(\boldsymbol{P}_j^*,\boldsymbol{Q}_j^*)$ is strictly positive. When $\theta_j$ approaches infinity, the $j$-th agent's objective functional $\mathbb{E}\phi_j[X_j(T)]-\theta_jD(\boldsymbol{P}_j^*,\boldsymbol{Q}_j^*)$ approaches negative infinity, which does not reach its supremum. Therefore, all $n$ agents' optimal strategies must converge to the same function, and we define it as the asymptotic strategy $\{\tilde{\boldsymbol{P}}(t)\}_{t\in\mathcal{T}}$.

Next, we derive the expression of the asymptotic strategy. 
Substituting \eqref{eq:opinion} into \eqref{eq:u}, using the L'Hôpital's rule, we have
\begin{equation}
    \lim_{\{\theta_j\}_{j\in\mathcal{N}}\to\infty}\boldsymbol{U}=\mathrm{diag}\left(\frac{\eta_i\alpha_i}{\sum_{j\in\mathcal{N}}\eta_j\alpha_j}\right)_{i\in\mathcal{N}}\otimes\boldsymbol{I}_m.
    \label{eq:tildeu}
\end{equation}
Therefore, substituting \eqref{eq:tildeu} into \eqref{eq:pu}, we can prove that
\begin{equation}
    \tilde{\boldsymbol{P}}(t)=\frac{\sum_{j\in\mathcal{N}}\eta_j}{\sum_{j\in\mathcal{N}}\eta_j\alpha_j}\mathrm{e}^{r(t-T)}\boldsymbol{\varSigma}^{-1}\boldsymbol{\nu},
\end{equation}
which is \eqref{eq:asymptotic} given the definition of $\tilde{\alpha}$ in \eqref{eq:asy_alpha}.
\end{proof}

From Theorem \ref{the:3}, to compare the $j$-th agent's rational strategy $\{\bar{\boldsymbol{P}}_j(t)\}_{t\in\mathcal{T}}$ and the asymptotic strategy $\{\tilde{\boldsymbol{P}}(t)\}_{t\in\mathcal{T}}$, we only need to compare the $j$-th agent's risk aversion coefficient $\alpha_j$ and the asymptotic risk aversion coefficient $\tilde{\alpha}$. From \eqref{eq:asy_alpha}, the asymptotic risk aversion coefficient $\tilde{\alpha}$ represents the average level of risk aversion among all agents. We define the \textit{social average agent} as a representative agent whose behavior reflects the average characteristics of a group of agents, and whose risk aversion coefficient is $\tilde{\alpha}$ in \eqref{eq:asy_alpha}. Therefore, when the influence coefficients $\{\theta_j\}_{j\in\mathcal{N}}$ approaches infinity, all $n$ agents eventually converge to the social average agent, sharing the same level of risk aversion.

If $\alpha_j<\tilde{\alpha}$, the absolute value of $\{\tilde{\boldsymbol{P}}(t)\}_{t\in\mathcal{T}}$ is smaller than that of $\{\bar{\boldsymbol{P}}_j(t)\}_{t\in\mathcal{T}}$. It indicates that if the $j$-th agent is more risk-taking than the social average agent, he/she will become more risk-averse under mutual influence. Conversely, if $\alpha_j>\tilde{\alpha}$, the absolute value of $\{\tilde{\boldsymbol{P}}(t)\}_{t\in\mathcal{T}}$ is larger than that of $\{\bar{\boldsymbol{P}}_j(t)\}_{t\in\mathcal{T}}$. It indicates that if the $j$-th agent is more risk-averse than the social average agent, he/she will become more risk-taking under mutual influence. 


\subsubsection{The Relationship Among Optimal Strategy, Rational Strategy, and Asymptotic Strategy}
With Theorem \ref{the:3}, we can further prove that the $j$-th agent's optimal strategy $\{\boldsymbol{P}_j^*(t)\}_{t\in\mathcal{T}}$ is a linear combination of his/her rational strategy $\{\bar{\boldsymbol{P}}_j(t)\}_{t\in\mathcal{T}}$ and the asymptotic strategy $\{\tilde{\boldsymbol{P}}(t)\}_{t\in\mathcal{T}}$. 

\begin{theorem}\label{the:5}
The $j$-th agent's optimal strategy $\{\boldsymbol{P}_j^*(t)\}_{t\in\mathcal{T}}$ in the differential game problem \eqref{eq:game_problem} is 
\begin{equation}
    \boldsymbol{P}_j^*(t)=\boldsymbol{V}_j\bar{\boldsymbol{P}}_j(t)+(\boldsymbol{I}_m-\boldsymbol{V}_j)\tilde{\boldsymbol{P}}(t),
    \label{eq:decomposition2}
\end{equation}
where the weight matrix is
\begin{equation}
    \boldsymbol{V}_j:=(\alpha_j^{-1}-\tilde{\alpha}^{-1})^{-1}\left(\sum_{i\in\mathcal{N}}\alpha_i^{-1}\boldsymbol{U}_{ji}-\tilde{\alpha}^{-1}\boldsymbol{I}_m\right).
    \label{eq:v}
\end{equation}
\end{theorem}

\begin{proof}
Combining \eqref{eq:rational}, \eqref{eq:decomposition}, and \eqref{eq:asymptotic}, we can obtain \eqref{eq:decomposition2} and \eqref{eq:v}, and we omit the details here.
\end{proof}

Theorem \ref{the:5} shows the relationship among the $j$-th agent's optimal strategy $\{\boldsymbol{P}_j^*(t)\}_{t\in\mathcal{T}}$, rational strategy $\{\bar{\boldsymbol{P}}_j(t)\}_{t\in\mathcal{T}}$, and the asymptotic strategy $\{\tilde{\boldsymbol{P}}(t)\}_{t\in\mathcal{T}}$. When $\theta_j=0$, we have the weight matrix $\boldsymbol{V}_j=\boldsymbol{I}_m$, which indicates that $\{\boldsymbol{P}_j^*(t)\}_{t\in\mathcal{T}}=\{\bar{\boldsymbol{P}}_j(t)\}_{t\in\mathcal{T}}$, i.e., the $j$-th agent makes rational strategies without others influence. As $\theta_j$ approaches infinity, we have the weight matrix $\boldsymbol{V}_j=\boldsymbol{O}_m$, which indicates that $\{\boldsymbol{P}_j^*(t)\}_{t\in\mathcal{T}}=\{\tilde{\boldsymbol{P}}(t)\}_{t\in\mathcal{T}}$, i.e., the $j$-th agent's optimal strategy converges to the asymptotic strategy under mutual influence among agents. As $\theta_j$ increases from $0$ to infinity, the weight matrix $\boldsymbol{V}_j$ transitions from $\boldsymbol{I}_m$ to $\boldsymbol{O}_m$, signifying that as the mutual influence becomes strong, the $j$-th agent's optimal strategy gradually shifts from his/her rational strategy to the asymptotic strategy. However, due to the complex form of \eqref{eq:v}, it is challenging to derive the quantitative relationship between $\theta_j$ and $\{\boldsymbol{P}_j^*(t)\}_{t\in\mathcal{T}}$, and we will use numerical experiments to study the impact of $\theta_j$ on $\{\boldsymbol{P}_j^*(t)\}_{t\in\mathcal{T}}$ in Section \ref{sec:experiment}.

\subsection{The Terminal Wealth}
Next, we analyze the impact of the mutual influence on agents' terminal wealth. We compare agents' terminal wealth when the influence coefficients $\{\theta_j\}_{j\in\mathcal{N}}$ are zero to the case when $\{\theta_j\}_{j\in\mathcal{N}}$ approach infinity. We will study the general cases with arbitrary $\{\theta_j\}_{j\in\mathcal{N}}$ using numerical experiments in Section \ref{sec:experiment}. When the influence coefficients $\{\theta_j\}_{j\in\mathcal{N}}$ are equal to zero and agents make independent strategies without mutual influence, we call the $j$-th agent's terminal wealth $\bar{X}_j(T)$ the \textit{rational terminal wealth}. As the influence coefficients $\{\theta_j\}_{j\in\mathcal{N}}$ approaches infinity, we call the $j$-th agent's terminal wealth $\tilde{X}_j(T)$ the \textit{asymptotic terminal wealth}. In Theorem \ref{the:4}, we derive the means and variances of the rational terminal wealth $\bar{X}_j(T)$ and asymptotic terminal wealth $\tilde{X}_j(T)$, respectively.

\begin{theorem}\label{the:4}
The means and variances of the $j$-th agent's rational terminal wealth $\bar{X}_j(T)$ and asymptotic terminal wealth $\tilde{X}_j(T)$ are
\begin{align}
    \mathbb{E}\bar{X}_j(T)=x_j\mathrm{e}^{rT}+\alpha_j^{-1}\kappa T\quad \text{and}\quad \ \mathbb{D}\bar{X}_j(T)=\alpha_j^{-2}\kappa T,\label{eq:ed1}\\
    \mathbb{E}\tilde{X}_j(T)=x_j\mathrm{e}^{rT}+\tilde{\alpha}^{-1}\kappa T\quad \text{and}\quad \mathbb{D}\tilde{X}_j(T)=\tilde{\alpha}^{-2}\kappa T.\label{eq:ed2}
\end{align}
\end{theorem}

\begin{proof}
Following our prior work in \cite{wang2024herd}, the mean and variance of the $j$-th agent's terminal wealth are
\begin{align}
    \mathbb{E}X_j(T)&=x_j\mathrm{e}^{rT}+\int_0^T\mathrm{e}^{r(T-t)}\boldsymbol{P}_j^\top(t)\boldsymbol{\nu}\dif t\quad \text{and}\label{eq:ex}\\
    \mathbb{D}X_j(T)&=\int_0^T\mathrm{e}^{2r(T-t)}\boldsymbol{P}_j^\top(t)\boldsymbol{\varSigma}\boldsymbol{P}_j(t)\dif t,\label{eq:dx}
\end{align}
respectively. Substituting \eqref{eq:rational} and \eqref{eq:asymptotic} into \eqref{eq:dx}, we can obtain \eqref{eq:ed1} and \eqref{eq:ed2}.
\end{proof}

From Theorem \ref{the:4}, to compare the $j$-th agent's rational terminal wealth $\bar{X}_j(T)$ and asymptotic terminal wealth $\tilde{X}_j(T)$, we only need to compare his/her risk aversion coefficient $\alpha_j$ and the asymptotic risk aversion coefficient $\tilde{\alpha}$. 
If $\alpha_j<\tilde{\alpha}$, the mean and variance of the asymptotic terminal wealth are smaller than those of the rational terminal wealth, i.e., $\mathbb{E}\tilde{X}_j(t)<\mathbb{E}\bar{X}_j(t)$ and $\mathbb{D}\tilde{X}_j(t)<\mathbb{D}\bar{X}_j(t)$. It indicates that if the $j$-th agent is more risk-taking than the social average agent, he/she will seek lower returns and risk under mutual influence. 
Conversely, if $\alpha_j>\tilde{\alpha}$, then $\mathbb{E}\tilde{X}_j(t)>\mathbb{E}\bar{X}_j(t)$ and $\mathbb{D}\tilde{X}_j(t)>\mathbb{D}\bar{X}_j(t)$. It indicates that if the $j$-th agent is more risk-averse than the social average agent, he/she will seek higher returns and risk under mutual influence. 

\section{Numerical Experiments}
\label{sec:experiment}
In this section, we conduct numerical experiments to validate the performance of the fast algorithm and to verify the correctness of our analysis of the impact of mutual influence on agents' optimal strategies and terminal wealth.

\subsection{The Accuracy and Efficiency of the Fast Algorithm}
\label{sec:experiment-A}
\subsubsection{Experiment setup}
We set the investment period as $T=50$. We use the 1-year fixed deposit interest rate of the Bank of China $r=1.45\%$ in the year 2023. We collect daily closing prices of $m=5$ stocks from 2019 to 2023 to determine the excess return rates $\boldsymbol{\nu}$ and covariance matrix $\boldsymbol{\varSigma}$ for the risky assets. Details are in \cite{sf}. We consider scenarios in the network with $n=10$, $50$, $100$, and $500$ agents. For each scenario, we randomly generate the adjacency matrix $\boldsymbol{W}$, and set the risk aversion coefficients $\{\alpha_j\}_{j\in\mathcal{N}}$, influence coefficients $\{\theta_j\}_{j\in\mathcal{N}}$, and initial wealth $\{x_j\}_{j\in\mathcal{N}}$. Following the works in \cite{yuen2001estimation, wang2024herd}, we assume that $\alpha_j\in[0.1,0.5]$, $\theta_j\in[10^{-5},10^{-4}]$, and $x_j\in[1,5]$ for all $j\in\mathcal{N}$. We observe the same trend for other values of the parameters.

\subsubsection{Experimental results}
We compare the accuracy and efficiency of the {\textbf{Fast}} algorithm in Algorithm \ref{alg:1}, with the {\textbf{Base}} algorithm that calculates the exact solutions of agents' optimal strategies using \texttt{scipy.integrate.quad} from Theorem \ref{the:1} and Theorem \ref{the:2}. We use the Relative Error of agents' optimal strategies to measure the algorithm accuracy, which is defined as
\begin{equation}
    \text{Relative Error}:=\frac{1}{n}\sum_{j\in\mathcal{N}}\frac{\Vert\boldsymbol{P}_j^*(t)-\hat{\boldsymbol{P}}_j^*(t)\Vert_2^2}{\Vert\boldsymbol{P}_j^*(t)\Vert_2^2},
\end{equation}
where $\Vert\boldsymbol{f}(t)\Vert_2^2:=\int_0^T(\boldsymbol{f}^\top\boldsymbol{f})(t)\dif t$, 
$\boldsymbol{P}_j^*(t)$ represents the $j$-th agent's exact optimal strategy calculated using {\textbf{Base}}, and $\hat{\boldsymbol{P}}_j^*(t)$ represents the approximated optimal strategies calculated using {\textbf{Fast}}. To validate the efficiency of the fast algorithm, we compare the computation times. 

The experimental results are in Table \ref{tab:2}. From Table \ref{tab:2}, the Relative Error is less than $10^{-31}$, and {\textbf{Fast}} reduces the computation times by $98\%$ compared to {\textbf{Base}}, which validate the accuracy and efficiency of our proposed fast algorithm.


\begin{table}[!t]
\setlength{\tabcolsep}{16pt}
\centering
\caption{Comparison of the relative errors and computation times between \textbf{Fast} and \textbf{Base}}
\label{tab:2}
\begin{tabular}{c|c|c|c}
\hline
\multirow{2}{*}{$n$} & \multirow{2}{*}{Relative Error} & \multicolumn{2}{c}{Computation Time} \\
\cline{3-4}
& & {\textbf{Base}} & {\textbf{Fast}} \\
\hline
$10$ & $2.76\times10^{-32}$ & $1.84$s & $0.04$s \\
$50$ & $1.48\times10^{-32}$ & $33.91$s & $0.67$s \\
$100$ & $1.94\times10^{-32}$ & $141.09$s & $2.39$s \\
$500$ & $2.34\times10^{-32}$ & $2641.51$s 
& $49.81$s \\
\hline
\end{tabular}
\vspace{-3mm}
\end{table}

\subsection{The Impact of Mutual Influence on Agents' Optimal Strategies and Terminal Wealth}
\label{sec:experiment-B}

\begin{figure*}[!t]
\centering
\subfloat[Asset 1 ($\boldsymbol{W}'$)]{\includegraphics[width=0.25\linewidth]{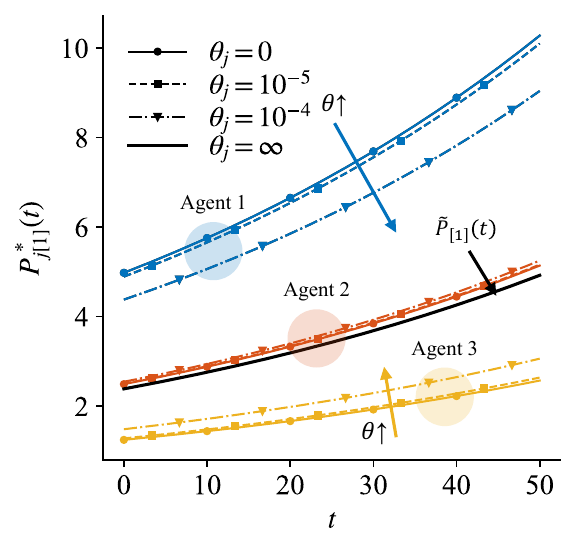}\label{fig:fig1a}}
\subfloat[Asset 2 ($\boldsymbol{W}'$)]{\includegraphics[width=0.25\linewidth]{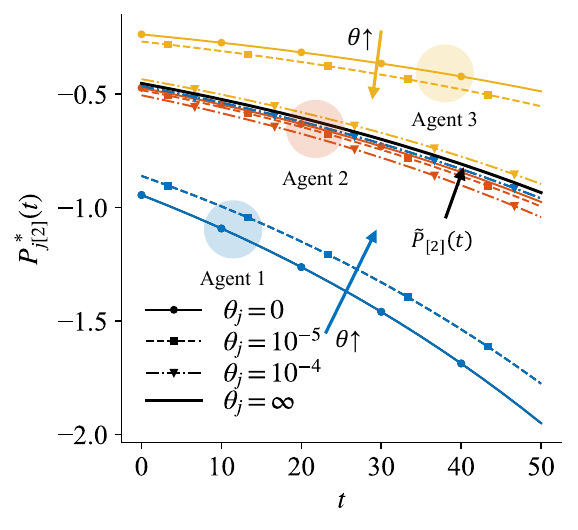}\label{fig:fig1b}}
\subfloat[Asset 1 ($\boldsymbol{W}''$)]{\includegraphics[width=0.25\linewidth]{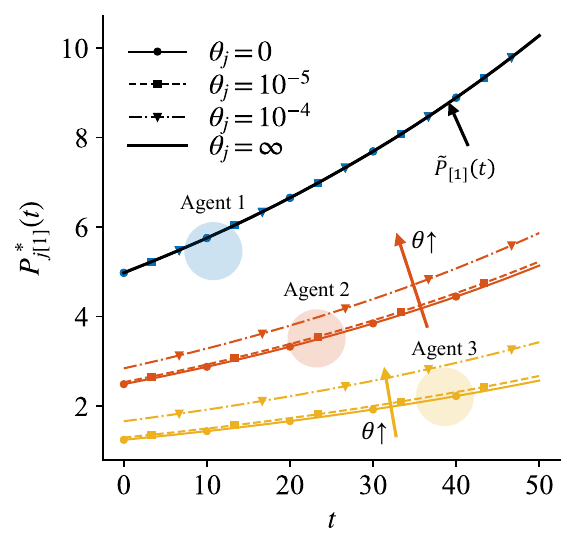}\label{fig:fig1c}}
\subfloat[Asset 2 ($\boldsymbol{W}''$)]{\includegraphics[width=0.25\linewidth]{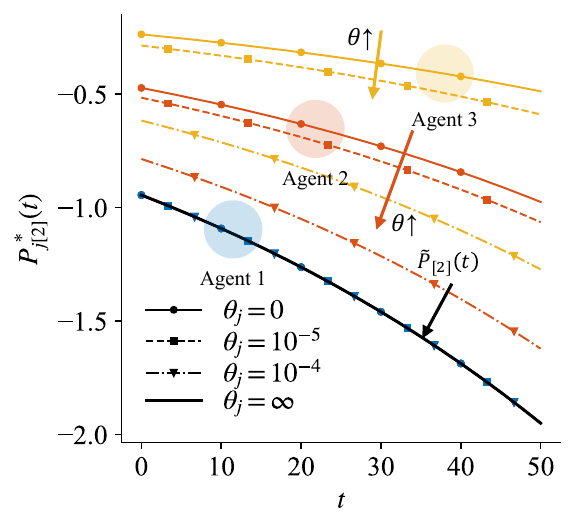}\label{fig:fig1d}}
\hfill
\subfloat[Mean ($\boldsymbol{W}'$)]{\includegraphics[width=0.25\linewidth]{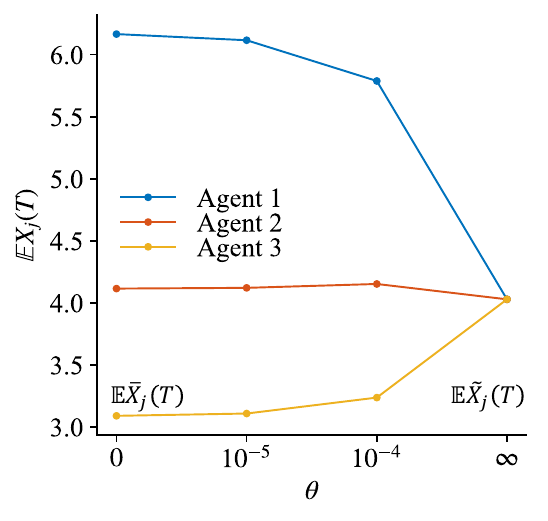}\label{fig:fig2a}}
\subfloat[Variance ($\boldsymbol{W}'$)]{\includegraphics[width=0.25\linewidth]{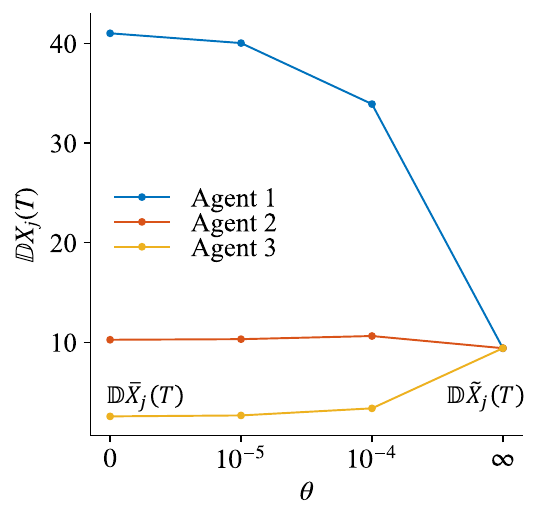}\label{fig:fig2b}}
\subfloat[Mean ($\boldsymbol{W}''$)]{\includegraphics[width=0.25\linewidth]{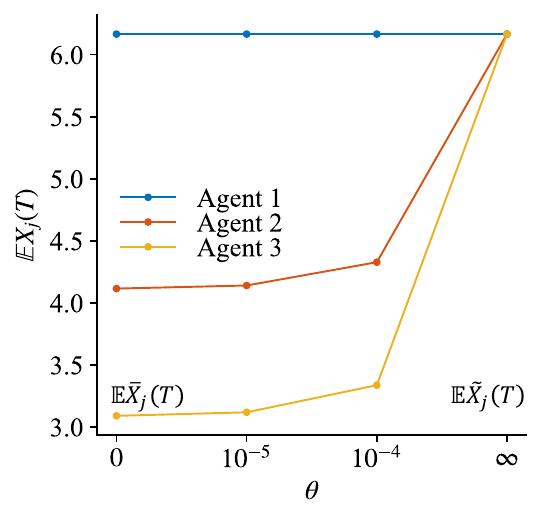}\label{fig:fig2c}}
\subfloat[Variance ($\boldsymbol{W}''$)]{\includegraphics[width=0.25\linewidth]{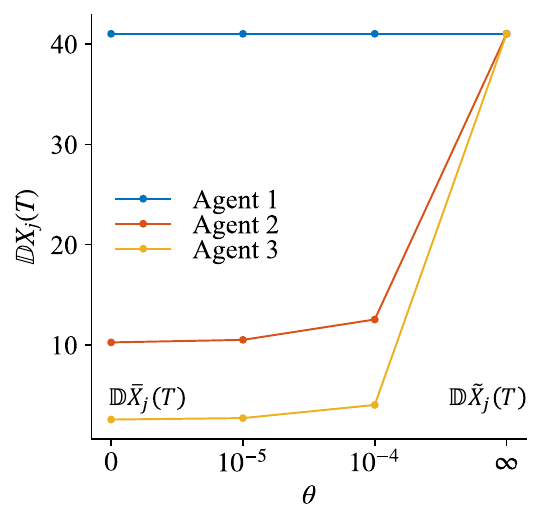}\label{fig:fig2d}}
\caption{The impact of mutual influence on agents' optimal strategies and terminal wealth.}\label{fig:fig1}
\vspace{-3mm}
\end{figure*}

\subsubsection{Experiment setup}
To facilitate the presentation of the results, we choose $n=3$ agents: Agent 1, 2, and 3, and $m=2$ stocks: Asset 1 and 2. Details are in \cite{sf}. We set the agents' risk aversion coefficients as $\alpha_1=0.1$, $\alpha_2=0.2$, and $\alpha_3=0.4$, and their initial wealth as $x_1=x_2=x_3=1$. We assume their influence coefficients are homogeneous, with values of $\theta=0$, $10^{-5}$, $10^{-4}$, and approaches infinity. We choose the adjacency matrices $\boldsymbol{W}'=\left[\begin{smallmatrix}
0 & 0.5 & 0.5 \\
0.5 & 0 & 0.5 \\
0.5 & 0.5 & 0
\end{smallmatrix}\right]$, which represent the cases where each agent is equally influenced by others, and $\boldsymbol{W}''=\left[\begin{smallmatrix}
1 & 0 & 0 \\
1 & 0 & 0 \\
1 & 0 & 0
\end{smallmatrix}\right]$, which represent the cases where each agent is only influenced by Agent 1, i.e., the leading expert \cite{wang2024herd}, when making strategies, as examples. By comparing the results with $\boldsymbol{W}'$ and $\boldsymbol{W}''$, we can analyze the differences in the impact of mutual influence and unilateral influence on agents' optimal strategies and terminal wealth. The other parameter settings are the same as Section \ref{sec:experiment-A}. We observe the same trend for other values of the parameters.

\subsubsection{Experimental results}
Next, we explain the results of agents' optimal strategies and terminal wealth, respectively.

\textbf{Optimal strategies}: We calculate agents' optimal strategies using Theorem \ref{the:1} and Theorem \ref{the:2}, and their rational strategies and asymptotic strategies using \eqref{eq:rational} and \eqref{eq:asymptotic}, respectively. The results are in Fig. \ref{fig:fig1a}--\ref{fig:fig1d}, where Fig. \ref{fig:fig1a}--\ref{fig:fig1b} plot agents' optimal strategies on Asset 1 and Asset 2 with adjacency matrix $\boldsymbol{W}'$, and Fig. \ref{fig:fig1c}--\ref{fig:fig1d} plot those with $\boldsymbol{W}''$, respectively. 

When the adjacency matrix is $\boldsymbol{W}'$, from Fig. \ref{fig:fig1a}--\ref{fig:fig1b}, as the influence coefficients increase, all three agents' optimal strategies converge to the asymptotic strategy, i.e., all three agents converge to the social average agent under mutual influence. We calculate the asymptotic risk aversion coefficient $\tilde{\alpha}'\approx0.21$ using \eqref{eq:asy_alpha}, which is larger than $\alpha_1$ and $\alpha_2$, and smaller than $\alpha_3$. From Fig. \ref{fig:fig1a}--\ref{fig:fig1b}, we can find that the absolute value of the asymptotic strategy $\{\tilde{\boldsymbol{P}}(t)\}_{t\in\mathcal{T}}$ is smaller than the rational strategies $\{\bar{\boldsymbol{P}}_1(t)\}_{t\in\mathcal{T}}$ and $\{\bar{\boldsymbol{P}}_2(t)\}_{t\in\mathcal{T}}$, and larger than $\{\bar{\boldsymbol{P}}_3(t)\}_{t\in\mathcal{T}}$. It validates Theorem \ref{the:3} that under mutual influence, risk-taking agents become more risk-averse, and vice versa. 
When the adjacency matrix is $\boldsymbol{W}''$, from Fig. \ref{fig:fig1c}--\ref{fig:fig1d}, as the influence coefficients increase, all three agents' optimal strategies converge to Agent 1's rational strategy, i.e., all of them make the same strategy as the leading expert under unilateral influence. 

\textbf{Terminal wealth}: We calculate the mean and variance of agents' rational and asymptotic terminal wealth, and their terminal wealth when $\theta=10^{-5}$ and $10^{-4}$ using \eqref{eq:wealth}. The results are in Fig \ref{fig:fig2a}--\ref{fig:fig2d}, where Fig. \ref{fig:fig2a}--\ref{fig:fig2b} plot the mean and variance of agents' terminal wealth with adjacency matrix $\boldsymbol{W}'$, and Fig. \ref{fig:fig2c}--\ref{fig:fig2d} with adjacency matrix $\boldsymbol{W}''$, respectively. 

When the adjacency matrix is $\boldsymbol{W}'$, from Fig. \ref{fig:fig2a}--\ref{fig:fig2b}, as the influence coefficients increase,  all three agents' terminal wealth converge to the asymptotic terminal wealth, i.e., all three agents converge to the social average agent under mutual influence. Because $\tilde{\alpha}'$ is larger than $\alpha_1$ and $\alpha_2$, and smaller than $\alpha_3$, the means and variances of Agent 1's and Agent 2's asymptotic terminal wealth are smaller than those of their rational terminal wealth, while the mean and variance of Agent 3's asymptotic terminal wealth are larger than those of his/her rational terminal wealth. It validates Theorem \ref{the:4} that under mutual influence, risk-taking agents seek lower returns and risk, and vice versa. 
When the adjacency matrix is $\boldsymbol{W}''$, from Fig. \ref{fig:fig2c}--\ref{fig:fig2d}, all three agents' terminal wealth converge to Agent 1's terminal wealth, i.e., all three agents obtain the same terminal wealth as that of the leading expert under unilateral influence. 

\section{Conclusion}
\label{sec:conclusion}
In this work, we study the optimal investment differential game problem with mutual influence in social networks and derive the analytical solutions for agents' optimal strategies. We propose a fast algorithm to find agents' approximated optimal strategies with low computational complexity. Numerical experiments show that the fast algorithm can reduce the computation time by $98\%$ while the relative error is no larger than $10^{-31}$. We theoretically analyze the impact of mutual influence on agents' optimal strategies and terminal wealth. When the mutual influence approaches infinity with homogeneous influence across all agents, we show that the agents' optimal strategies converge to the asymptotic strategy, which is the social average agent's rational strategy. Our analysis and numerical experiments show that if an agent's risk aversion coefficient is smaller than the asymptotic risk aversion coefficient, he/she will become more risk-averse and seek lower returns and risk under the mutual influence when compared to the social average agent, and vice versa.

\bibliographystyle{IEEEbib}
\bibliography{refs}

\end{document}